\documentclass[journal,comsoc]{IEEEtran}
\usepackage[T1]{fontenc}

\usepackage{amssymb,amsmath,amsfonts,amsthm}
\usepackage{mathrsfs}
\usepackage{cite}
\usepackage{array}
\usepackage{tabularx}
\usepackage[table]{xcolor}
\usepackage{color}
\usepackage{mathtools}
\usepackage{subfigure}
\usepackage{textcomp}
\usepackage{gensymb}
\usepackage{mathtools}
\usepackage{graphicx}
\usepackage{bbm}
\usepackage[utf8]{inputenc}

\newtheorem{fact}{Fact}

\newtheorem{lemma}{Lemma}

\newtheorem{propi}{Proposition}

\IEEEoverridecommandlockouts
\ifCLASSINFOpdf
\else
\fi
\usepackage[cmintegrals]{newtxmath}

\begin{document}
	\title{When is the Achievable Rate Region Convex in Two-User Massive MIMO Systems?}
\author{Zheng~Chen, Emil~Bj\"{o}rnson, and~Erik~G.~Larsson
	\thanks{Z.~Chen, E.~Bj\"{o}rnson, and E.~G.~Larsson are with the Department of Electrical Engineering (ISY), Link\"{o}ping University,  58183 Link\"{o}ping, Sweden (email: \{zheng.chen, emil.bjornson, erik.g.larsson\}@liu.se).}%
	\thanks{This work was supported in part by ELLIIT, CENIIT, and the Swedish Foundation for Strategic Research (SSF).} \vspace{-6mm}}
\maketitle

\begin{abstract}
This letter investigates the achievable rate region in Massive multiple-input-multiple-output (MIMO) systems with two users, with focus on the i.i.d.~Rayleigh fading and line-of-sight (LoS) scenarios. If the rate region is convex, spatial multiplexing is preferable to orthogonal scheduling, while the opposite is true for non-convex regions. We prove that the uplink and downlink rate regions with i.i.d.~Rayleigh fading are convex, while the convexity in LoS depends on  parameters such as angular user separation, number of antennas, and signal-to-noise ratio (SNR).
\end{abstract}
\vspace{-0.5cm}
\begin{keywords}
Rate region, convexity, Massive MIMO.
\end{keywords}

\vspace{-0.3cm}
\section{Introduction}

The achievable rate region of a multi-user multiple-input multiple-output (MIMO) system can be non-convex when using suboptimal transmission schemes that treat interference as noise or utilize imperfect channel state information (CSI) \cite{caire-achievable-mimo, ergodic-mimo-broadcast,bjornson2013optimal}. Yet, the capacity region is always convex \cite{capacity-mimo-broadcast}, since the capacity-achieving scheme may implicitly use time-sharing (scheduling) to convexify a region by operating on a line between two achievable points. In practice, it is important to know if a scheme for broadcast or multiple access channels must be accompanied with scheduling to maximize the rates.

Massive MIMO refers to a multi-user MIMO system with  a very large number of service antennas and this approach is a key enabler for the next generation of wireless networks \cite{zhang-overview}. Despite the vast research on Massive MIMO, prior works have not explicitly targeted the convexity of the rate region.
The common practice in Massive MIMO is to serve all users by spatial multiplexing \cite{massivemimobook}, although that is only preferable when the rate region is convex. In contrast, when the rate region is non-convex, higher sum throughput can be achieved by scheduling some at a time. In the basic two-user case, scheduling can be implemented by TDMA/FDMA/OFDMA.

In this letter, we characterize the achievable rate region of two-user Massive MIMO systems to gain fundamental insights into the convexity properties. A general form of rate expressions is considered and the Pareto boundaries in uplink (UL) and downlink (DL) are derived. By defining the boundary curve by parametric equations of the power-control coefficients, we derive necessary and sufficient conditions for convexity. In two specific cases, namely i.i.d.~Rayleigh fading and line-of-sight (LoS) channels, the exact convexity conditions are derived, discussed, and illustrated numerically.

 Note that the analysis in this letter is not limited to Massive MIMO systems. It can also capture other general cases when a single-antenna transmitter intends to send different signals to two receivers, such as a broadcast channel.

\vspace{-0.2cm}
\section{Two-user Rate Region in Massive MIMO}
\label{sec:general_downlink}
Consider a Massive MIMO system where a base station (BS) equipped with $M$ antennas serves two users. The achievable rates can be written in the general form
\vspace{-0.1cm}
\begin{equation}
R_1=\log_2\left(1+\frac{\alpha_1\eta_1}{\mu_{1,1} \eta_1+\mu_{1,2} \eta_2+1}\right),\label{rate_1}
\end{equation}
\begin{equation}
R_2=\log_2\left(1+\frac{\alpha_2 \eta_2}{\mu_{2,1} \eta_1+\mu_{2,2} \eta_2+1}\right), \label{rate_2}
\end{equation}
where $\eta_k$ is the power-control coefficient and $\alpha_k\geq0$ is the effective channel gain of user $k$, for $k=1,2$. The noise power is normalized to 1, $\mu_{1,1}\geq0$ and $\mu_{2,2}\geq0$ are the self-interference coefficients caused by having imperfect CSI, and $\mu_{1,2}\geq0$ and $\mu_{2,1}\geq0$ are the inter-user interference coefficients. 
Note that $\eta_1$ and $\eta_2$ are design variables.
In the DL, the power constraint at the BS is $0\!\leq\! \eta_{1}+\eta_{2}\!\leq \!1$. In the UL, the power constraints of the two users are $0\leq\eta_1\leq 1$ and $0\leq\eta_2\leq1$. We keep the coefficients arbitrary in this section, while specific expressions are given in Sections~\ref{ergodic} and \ref{LoS}.

Using a similar definition of the Pareto boundary as in \cite{capacity-region}, 
in the following lemma, we state the general interpretation of convexity of the achievable rate region.
\begin{lemma}
	\label{lemma-convexity}
		Denote by $R_{1,\textnormal{bd}}$ and $R_{2,\textnormal{bd}}$ the achievable rates of user $1$ and user $2$ at the Pareto boundary. Both $R_{1,\textnormal{bd}}$ and $R_{2,\textnormal{bd}}$ are functions of the power-control coefficients $\eta_{1}$ and $\eta_{2}$. The achievable rate region is convex, if and only if $R_{1,\textnormal{bd}}$ is a concave function of $R_{2,\textnormal{bd}}$ for $R_{2,\textnormal{bd}}\in[0, R_{2,\max}]$, where $R_{2,\max}$ is the maximum rate of user $2$.
		If $R_{1,\textnormal{bd}}(R_{2,\textnormal{bd}})$  is continuous and twice differentiable for $R_{2,\textnormal{bd}}\in[0, R_{2,\max}]$, then $R_{1,\textnormal{bd}}(R_{2,\textnormal{bd}})$ is a concave function if and only if $\frac{d^2 R_{1,\textnormal{bd}}}{d R_{2,\textnormal{bd}}^2} \leq 0$.
\end{lemma}

\vspace{-0.4cm}
\subsection{Convexity of the DL Rate Region}
In the DL, the rates in \eqref{rate_1} and \eqref{rate_2} can be further simplified if we only consider the rates at the Pareto boundary.
\begin{fact}
	At the Pareto boundary $R_{1,\textnormal{bd}}(R_{2,\textnormal{bd}})$ of the DL rate region, we have $\eta_1+\eta_2=1$, $\eta_{1}\in[0,1]$ \cite[Theorem~1.9]{bjornson2013optimal}. 
\end{fact}
\vspace{-0.1cm}
Based on this fact, the user rates at the Pareto boundary can be given as functions of a single variable $\eta_1$:
\vspace{-0.1cm}
\begin{align}
R_{1,\textnormal{bd}}(\eta_{1})&=\log_2\left(1+\frac{\alpha_1\eta_1}{(\mu_{1,1}-\mu_{1,2}) \eta_1+\mu_{1,2}+1}\right), \label{eq:rate1_general}\\
R_{2,\textnormal{bd}}(\eta_{1})&=\log_2\left(1+\frac{\alpha_2 (1-\eta_1)}{(\mu_{2,1} -\mu_{2,2} )\eta_1+\mu_{2,2}+1}\right). \label{eq:rate2_general}
\end{align}
Both $R_{1,\textnormal{bd}}$ and $R_{2,\textnormal{bd}}$ are continuous and twice-differentiable functions of $\eta_1$. The second-order derivative $\frac{d^2 R_{1,\textnormal{bd}}}{d R_{2,\textnormal{bd}}^2}$ is obtained by utilizing following fact.
\begin{fact}
	\label{def:chain_rule}
	Consider a two-dimensional curve with coordinates $(x,y)$ defined parametrically as $x=g(u)$ and $y=f(u)$, where $g$ and $f$ are continuous and twice-differentiable functions of $u$. The first-order derivative of $y$ with respect to $x$ is
		\vspace{-0.1cm}
	\begin{equation}
	\frac{d y}{d x}\mathop{=}\limits^{(a)}\frac{d y}{d u}\frac{d u}{d x}\mathop{=}\limits^{(b)}\frac{\frac{d y}{d u}}{\frac{d x}{d u}}=\frac{ f'(u)}{g'(u)},
	\label{eq:chain-rule-one}
		\vspace{-0.1cm}
	\end{equation}
	where $(a)$ is the chain rule in Leibniz's notation and $(b)$ follows from the inverse function rule.
	The second-order derivative of $y$ with respect to $x$ is
		\vspace{-0.2cm}
	\begin{align}
	\frac{d^2 y}{d x^2}&=\frac{d}{d u}\left(\frac{d y}{d x}\right)\cdot \frac{d u}{d x} = \frac{d \left(\frac{ f'(u)}{g'(u)}\right)}{d u}\cdot \frac{1}{g'(u)}   \nonumber \\
	&=\frac{f''(u)g'(u)-f'(u)g''(u)}{\left(g'(u)\right)^3}.
	\label{eq:chain-rule-second}
	\vspace{-0.3cm}
	\end{align}
	Note that $\frac{d^2 y}{d x^2}$ is obtained as a function of $u$. 
\end{fact}

From \eqref{eq:rate1_general} and \eqref{eq:rate2_general}, using Lemma~\ref{lemma-convexity}, it is straightforward to obtain an exact convexity condition, but the expression is large and depends on six coefficients, thus it provides little insights.\footnote{The interested reader can find the exact condition in the supplementary document provided along with this paper.} However, for specific values of $\alpha_1, \alpha_2$ and $\mu_{1,1}, \mu_{1,2}, \mu_{2,1}, \mu_{2,2}$, that expression can be used to validate if $\frac{d^2 R_{1,\textnormal{bd}}}{d R_{2,\textnormal{bd}}^2} \leq 0$.

\begin{figure}[ht!]
	\vspace{-0.3cm}
	\centering
	\includegraphics[trim={1cm 2.5cm  1cm 3.2cm},clip, width=1.0\columnwidth]{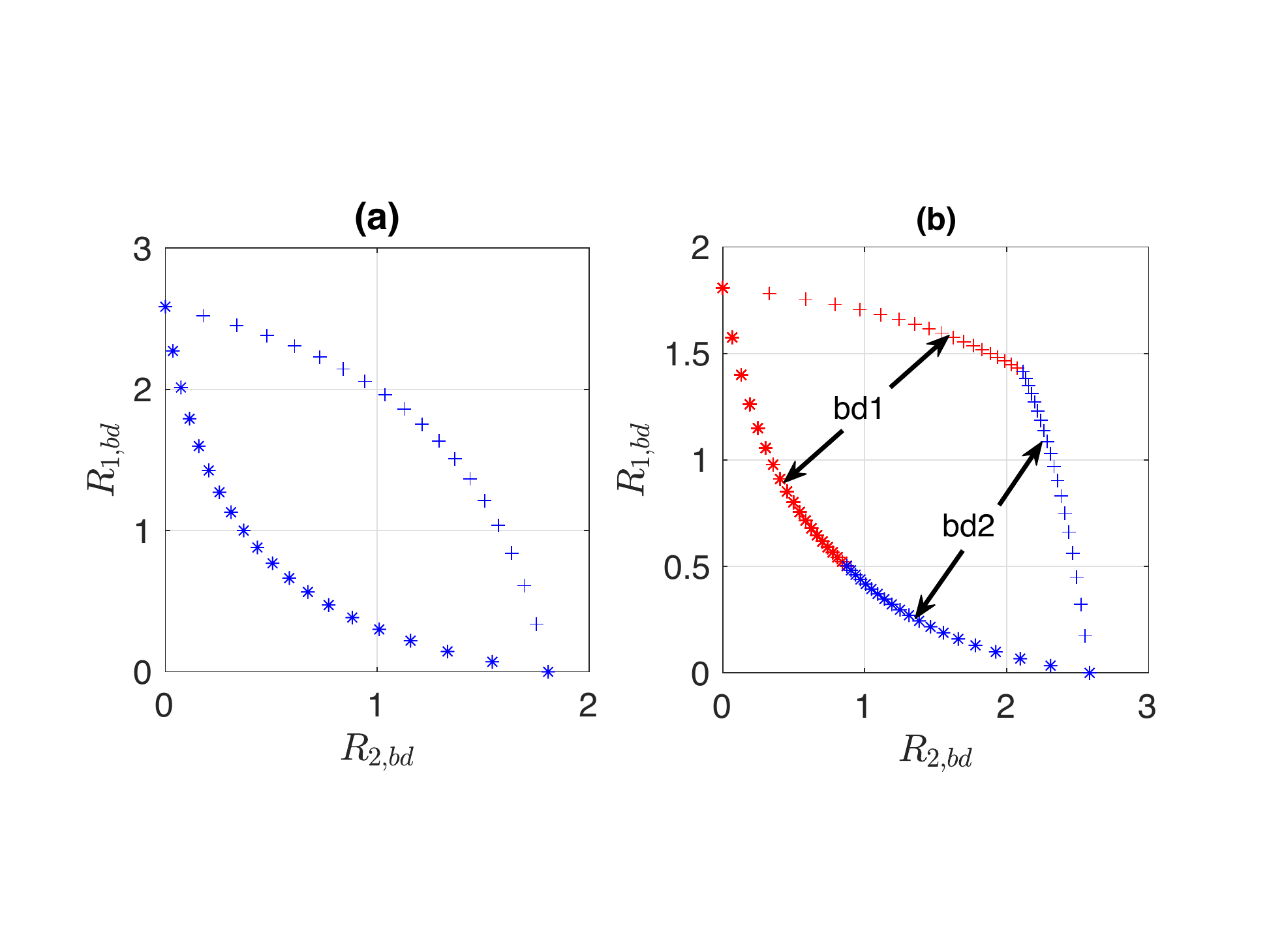}
	\vspace{-0.8cm}
	\caption{DL and UL Pareto boundaries with $\alpha_1=5$, $\alpha_2=10$, $\mu_{1,1}=\mu_{2,2}=1$. (a): The DL rate region is convex (marked with crosses) with $\mu_{1,2}=\mu_{2,1}=1$, and non-convex (marked with stars) with $\mu_{1,2}\!=\!\mu_{2,1}\!=\!10$; (b) The UL rate region is convex with $\mu_{1,2}\!=\!\mu_{2,1}\!=\!1$, and non-convex with $\mu_{1,2}\!=\!\mu_{2,1}\!=\!10$. }
	\label{fig:rate_downlink}
\end{figure}

Fig.~\ref{fig:rate_downlink}(a) shows examples of non-convex and convex DL rate regions. We see that when the inter-user interference coefficients are much larger than the self-interference coefficients, the achievable rate region has a non-convex shape.

\vspace{-0.3cm}
\subsection{Convexity of the UL Rate Region}

We now shift focus to the Pareto boundary in the UL. 
\begin{fact}
	\label{fact:2}
	The Pareto boundary $R_{1,\textnormal{bd}}(R_{2,\textnormal{bd}})$ of the UL rate region consists of two segments: one with $\eta_1=1$ and any $\eta_2\in[0,1]$, the other one with $\eta_2=1$ and any $\eta_1\in[0,1]$. 
\end{fact}
\begin{proof}
	For any $\eta_{1},\eta_{2}\in(0,1)$, by scaling both coefficients with $\delta>1$ such that either $\delta\eta_{1}=1$ or $\delta \eta_{2}=1$ is reached, a point at the Pareto boundary is achieved \cite[Theorem~1.9]{bjornson2013optimal}.
\end{proof}
At the first segment of the Pareto boundary with $\eta_1=1$, the rates can be simplified as functions of $\eta_2$:
\vspace{-0.15cm}
\begin{equation}
R_{1,\textnormal{bd1}}(\eta_{2})=\log_2\left(1+\frac{\alpha_1}{\mu_{1,2} \eta_2+\mu_{1,1} +1}\right),
\end{equation}
\vspace{-0.3cm}
\begin{equation}
R_{2,\textnormal{bd1}}(\eta_{2})=\log_2\left(1+\frac{\alpha_2 \eta_2}{\mu_{2,2} \eta_2+\mu_{2,1} +1}\right).
\end{equation}
Similarly, the second segment of the Pareto boundary is with $\eta_2=1$ and we can write the rates as functions of $\eta_1$:
\vspace{-0.15cm}
\begin{eqnarray}
R_{1,\textnormal{bd2}}(\eta_{1})&=&\log_2\left(1+\frac{\alpha_1\eta_1}{\mu_{1,1} \eta_1+\mu_{1,2} +1}\right),\\
R_{2,\textnormal{bd2}}(\eta_{1})&=&\log_2\left(1+\frac{\alpha_2 }{\mu_{2,1} \eta_1+\mu_{2,2} +1}\right).
\end{eqnarray}

Similar to the DL,  $R_{1,\textnormal{bd1}}(R_{2,\textnormal{bd1}})$ and $R_{1,\textnormal{bd2}}(R_{2,\textnormal{bd2}})$ need to be concave functions if the rate region is convex.
An additional condition at the interconnecting point of two segments of the boundary curve must also be satisfied.

\begin{lemma}
The two-user UL achievable rate region is convex if and only if (1) $\frac{d^2 R_{1,\textnormal{bd1}}}{d R_{2,\textnormal{bd1}}^2} \leq 0$; and (2) $\frac{d^2 R_{1,\textnormal{bd2}}}{d R_{2,\textnormal{bd2}}^2} \leq 0$; and (3) $\frac{R'_{1,\textnormal{bd1}}(\eta_{2})}{R'_{2,\textnormal{bd1}}(\eta_{2})}\Bigm|_{\eta_{2}=1}\geq \frac{R'_{1,\textnormal{bd2}}(\eta_{1})}{R'_{2,\textnormal{bd2}}(\eta_{1})}\Bigm|_{\eta_{1}=1}$.
\label{lemma-ul}
\end{lemma}
\begin{proof}
	The first two conditions follow directly from Lemma~\ref{lemma-convexity}. Since the Pareto boundary consists of two segments, at the interconnecting point with $\eta_{1}=\eta_{2}=1$, the boundary curve is not twice-differentiable. Thus, $\frac{d R_{1,\text{bd1}}}{d R_{2,\text{bd1}}}\Bigm|_{\eta_{2}=1}\geq\frac{d R_{1,\text{bd2}}}{d R_{2,\text{bd2}}}\Bigm|_{\eta_{1}=1}$ must be satisfied to assure the convexity of the region. From \eqref{eq:chain-rule-one}, we obtain the third condition in Lemma~\ref{lemma-ul}.
\end{proof}
\vspace{-0.1cm}
In Fig.~\ref{fig:rate_downlink}(b), we show examples of non-convex and convex UL rate regions.
We see that the UL Pareto boundary consists of two segments: $\small{\textsf{bd1}}$ with $\eta_1=1$ and $\small{\textsf{bd2}}$ with $\eta_2=1$. The combination of these segments completes the Pareto boundary.

\vspace{-0.2cm}
\section{Massive MIMO with i.i.d.~Rayleigh fading}
\label{ergodic}
As shown above, the convexity of the achievable rate region is strongly affected by the relations between the coefficients in the rate expressions. In this section, we analyze the rate region convexity for a Massive MIMO system with ergodic  i.i.d.~Rayleigh fading channels and single-antenna users.
\vspace{-0.2cm}
\subsection{Convexity of the DL Rate Region}
Denote by $\rho_{\textnormal{dl}}$ the total DL transmit power.
When using maximum ratio (MR) precoding based on imperfect channel knowledge as in \cite[Ch.~3]{marzetta2016fundamentals}, we have that a DL ergodic rate of user $k$ is  $R_k=\log_2\bigg(1+\frac{M\rho_{\textnormal{dl}} \gamma_k \eta_k}{1+\beta_k\rho_{\textnormal{dl}}(\eta_{1}+\eta_{2})}\bigg)$,
where $\eta_k$ is the power-control coefficient of user $k=1,2$ with $\eta_{1}+\eta_{2}\leq1$, $\beta_k$ is the channel gain of user $k$, and $\gamma_k$ is the mean square of the channel estimate.\footnote{With zero-forcing (ZF) precoding, the rate expressions are different but the same approach can be taken and the conclusions are the same.}

In terms of the general form in \eqref{rate_1} and \eqref{rate_2}, we have $\alpha_1=M\rho_{\textnormal{dl}} \gamma_1$, $\alpha_2=M\rho_{\textnormal{dl}} \gamma_2$, $\mu_{1,1}=\mu_{1,2}=\beta_1\rho_{\textnormal{dl}}$ and $\mu_{2,1}=\mu_{2,2}=\beta_2\rho_{\textnormal{dl}}$.
At the Pareto boundary, since $\eta_1+\eta_2=1$, we rewrite the ergodic achievable rates $R_{1,\textnormal{bd}}$ and $R_{2,\textnormal{bd}}$ as functions of $\eta_1$:
\vspace{-0.15cm}
\begin{equation}
R_{1,\textnormal{bd}}\!=\!\log_2\!\left(1+\frac{\alpha_1 \eta_1}{1+\mu_{1,1}}\right)\!, \,\, 
R_{2,\textnormal{bd}}\!=\!\log_2\!\left(1+\frac{\alpha_2 (1-\eta_1)}{1+\mu_{2,2}}\right).
\end{equation}
Defining $\zeta = \frac{1}{\ln(2)}$, the first- and second-order derivatives are 
\vspace{-0.1cm}
\begin{equation}
R'_{1,\textnormal{bd}}(\eta_1)=\frac{\zeta \alpha_1 }{1+\mu_{1,1}+\alpha_1 \eta_1}>0, 
\vspace{-0.25cm}
\end{equation}
\begin{align}
R'_{2,\textnormal{bd}}(\eta_1)&=-\frac{\zeta \alpha_2}{1+\mu_{2,2}+\alpha_2-\alpha_2\eta_1}<0,\\
\vspace{-0.2cm}
R''_{1,\textnormal{bd}}(\eta_1)&=-\frac{\zeta \alpha_1^2}{(1+\mu_{1,1}+\alpha_1 \eta_1)^2}<0,  \\
\vspace{-0.2cm}
R''_{2,\textnormal{bd}}(\eta_1)&=-\frac{\zeta \alpha_2^2}{(1+\mu_{2,2}+\alpha_2-\alpha_2\eta_1)^2}<0.
\vspace{-0.2cm}
\end{align}
Then, using Lemma \ref{def:chain_rule}, it follows that
\vspace{-0.1cm}
\begin{equation}
\frac{d^2 R_{1,\textnormal{bd}} }{d R_{2,\textnormal{bd}}^2}=\frac{R''_{1,\textnormal{bd}}(\eta_1)R'_{2,\textnormal{bd}}(\eta_1)-R''_{2,\textnormal{bd}}(\eta_1)R'_{1,\textnormal{bd}}(\eta_1)}{\big(R'_{2,\textnormal{bd}}(\eta_1)\big)^3}<0.
\vspace{-0.15cm}
\end{equation}
This means that $R_{1,\textnormal{bd}}$ is a strictly concave function of $R_{2,\textnormal{bd}}$. From Lemma~\ref{lemma-convexity}, the achievable rate region is always convex.

\vspace{-0.2cm}
\subsection{Convexity of the UL Rate Region}
When using MR combining, an UL ergodic rate of user $k$ is $R_k=\log_2\bigg(1+\frac{M \rho_{\textnormal{ul}} \gamma_k \eta_k}{1+\sum_{k=1}^{2} \beta_k\rho_{\textnormal{ul}}\eta_{k}}\bigg)$ \cite[Ch.~3]{marzetta2016fundamentals}.
In terms of the general form in \eqref{rate_1} and \eqref{rate_2}, we have $\alpha_1=M\rho_{\textnormal{ul}}\gamma_1$, $\alpha_2=M\rho_{\textnormal{ul}}\gamma_2$, $\mu_{1,1}=\mu_{2,1}=\beta_1\rho_{\textnormal{ul}}$ and $\mu_{1,2}=\mu_{2,2}=\beta_2\rho_{\textnormal{ul}}$.
For the  Pareto boundary segment with $\eta_1=1$, we obtain
\begin{eqnarray}
\vspace{-0.2cm}
R_{1,\textnormal{bd1}}&=&\log_2\left(1+\frac{\alpha_1 }{\mu_{1,1}+\mu_{1,2}\eta_{2}+1 }\right), \\
R_{2,\textnormal{bd1}}&=&\log_2\left(1+\frac{\alpha_2\eta_{2} }{\mu_{1,1} +\mu_{1,2}\eta_{2} +1}\right).
\vspace{-0.2cm}
\end{eqnarray}
The first- and second-order derivatives of $R_{1,\textnormal{bd1}}$ and $R_{2,\textnormal{bd1}}$ are
\begin{equation}
\vspace{-0.15cm}
R'_{1,\textnormal{bd1}}(\eta_2)=-\frac{\zeta\alpha_1\mu_{1,2}}{\left(1+\mu_{1,1}+\mu_{1,2}\eta_{2}\right)\left(1+\mu_{1,1}+\alpha_1+\mu_{1,2} \eta_1\right)}, 
\end{equation}
\vspace{-0.3cm}
\begin{equation}
R'_{2,\textnormal{bd1}}(\eta_2)=\frac{\zeta\alpha_2 (\mu_{1,1}+1)}{\left(1+\mu_{1,1} +\mu_{1,2}\eta_{2} \right)\left(1+\mu_{1,1}+(\alpha_2+\mu_{1,2})\eta_2\right)},
\end{equation}
\vspace{-0.3cm}
\begin{equation}
R''_{1,\textnormal{bd1}}(\eta_2)=\frac{\zeta \alpha_1 \mu_{1,2}^2\left[\alpha_1+2(1+\mu_{1,1}+\mu_{1,2}\eta_{2})\right]}{\left(1+\mu_{1,1}+\mu_{1,2}\eta_{2}\right)^2\left(1+\mu_{1,1}+\alpha_1+\mu_{1,2} \eta_2\right)^2},
\end{equation}
\vspace{-0.3cm}
\begin{equation}
\begin{split}
&R''_{2,\textnormal{bd1}}(\eta_2)\!=\\&\!-\frac{\zeta\alpha_2 (\mu_{1,1}\!+\!1)\!\left[(1+\mu_{1,1})(2\mu_{1,2}+\alpha_2)+2\mu_{1,2}\eta_{2}(\mu_{1,2}+\alpha_2)\right]}{\left(1+\mu_{1,1} +\mu_{1,2}\eta_{2} \right)^2\left(1+\mu_{1,1}+(\alpha_2+\mu_{1,2})\eta_2 \right)^2}.  
\end{split}
\end{equation}
\vspace{-0.1cm}
Using \eqref{eq:chain-rule-second} in Lemma \ref{def:chain_rule}, we can show that
\begin{equation}
\begin{split} \label{eq:UL_derivation_iid}
\frac{d^2 R_{1,\textnormal{bd1}}}{d R_{2,\textnormal{bd1}}^2}=&\frac{R''_{1,\textnormal{bd1}}(\eta_2)R'_{2,\textnormal{bd1}}(\eta_2)-R''_{2,\textnormal{bd1}}(\eta_2)R'_{1,\textnormal{bd1}}(\eta_2)}{\big(R'_{2,\textnormal{bd1}}(\eta_2)\big)^3}\\
=&-\frac{\alpha_1\mu_{1,2}\left[\alpha_2(1+\alpha_1+\mu_{1,1})+\alpha_1\mu_{1,2}\right]}{\zeta\alpha_2^2(1+\mu_{1,1})^2(1+\alpha_1+\mu_{1,1}+\mu_{1,2}\eta_{2})^2}\\
&\cdot (1\!+\!\mu_{1,1}\!+\!\mu_{1,2}\eta_{2})(1\!+\!\mu_{1,1}\!+\!\alpha_2\eta_{2}\!+\!\mu_{1,2}\eta_{2})\\
<&0,
\end{split}
\vspace{-0.5cm}
\end{equation}
where the intermediate steps are omitted due to lack of space.\footnote{The Mathematica scripts used by the authors to generate the results are included as supplementary documents to this paper.} 
From \eqref{eq:UL_derivation_iid} it follows that
the function $R_{1,\textnormal{bd1}}(R_{2,\textnormal{bd1}})$ is always concave. Analogously, we can prove that at the second segment of the Pareto boundary with $\eta_2=1$, 
 $R_{1,\textnormal{bd2}}(R_{2,\textnormal{bd2}})$ is also concave. 
At the interconnecting point of the two segments, 
in order to have $\frac{R'_{1,\textnormal{bd1}}(\eta_{2})}{R'_{2,\textnormal{bd1}}(\eta_{2})}\Bigm|_{\eta_{2}=1}\geq \frac{R'_{1,\textnormal{bd2}}(\eta_{1})}{R'_{2,\textnormal{bd2}}(\eta_{1})}\Bigm|_{\eta_{1}=1}$, after lengthy derivations, we obtain $
\frac{\mu_{1,2}}{1+\mu_{1,1}}\leq \frac{1+\mu_{1,2}}{\mu_{1,1}}$,
which always holds. Thus, all the convexity conditions in Lemma~\ref{lemma-ul} are satisfied, which means that the rate region is always convex.
\vspace{-0.2cm}
\begin{propi}
	\label{propo1}
 In a two-user Massive MIMO system with ergodic i.i.d.~Rayleigh fading channels and single-antenna users, both the UL and the DL rate regions are convex.
\end{propi} 

\vspace{-0.3cm}
\section{Massive MIMO with LoS Propagation}
\label{LoS}
In Massive MIMO with LoS propagation, the channel is deterministic and can be estimated with a negligible overhead. We thus consider perfect CSI and no self-interference terms.
\vspace{-0.4cm}
\subsection{Convexity of the DL Rate Region}
We consider a BS equipped with a uniform linear array with $d_{\textnormal{H}}\leq1/2$ as antenna spacing (in wavelengths). Denote by $\theta_1$ and $\theta_2$ the angles of the two users as seen from the BS, 
the achievable DL rates with MR precoding are 
$R_{k}\!=\!\log_2\!\left(\!1+\!\frac{M\rho_{\textnormal{dl}}\beta_k\eta_k}{1+g(\theta_k, \theta_j)\rho_{\textnormal{dl}}\beta_k\eta_j}\right)$ for $k,j=\{1,2\}$, $j\neq k$ \cite{massivemimobook}, where 
\begin{equation}
 g(\theta_k, \theta_j)=\begin{cases} 
\frac{\sin^2\left[\pi d_{\textnormal{H}}M(\sin(\theta_k)-\sin(\theta_j))\right]}{M\sin^2\left[\pi d_{\textnormal{H}}(\sin(\theta_k)-\sin(\theta_j))\right]} &\text{if}~\sin(\theta_k)\neq\sin(\theta_j)\\
M & \text{if}~\sin(\theta_k)=\sin(\theta_j).
\end{cases}
\label{eq:gtheta}
\end{equation}
At the Pareto boundary, we have
\begin{equation}
R_{1,\textnormal{bd}}\!=\!\log_2\!\left(\!1+\!\frac{\alpha_1\eta_1}{1\!+\!\mu_{1,2}(1\!-\!\eta_{1})}\right),\,\, R_{2,\textnormal{bd}}\!=\!\log_2\!\left(\!1+\!\frac{\alpha_2(1-\eta_1)}{1\!+\!\mu_{2,1}\eta_{1}}\right)\!,\\
\end{equation}
where $\alpha_1=M\rho_{\textnormal{dl}}\beta_1$, $\alpha_2=M\rho_{\textnormal{dl}}\beta_2$, $\mu_{1,2}=g(\theta_1, \theta_2)\rho_{\textnormal{dl}}\beta_1$, $\mu_{2,1}=g(\theta_1, \theta_2)\rho_{\textnormal{dl}}\beta_2$.
From Lemma~\ref{lemma-convexity}, by calculating the second-order derivative, it is obvious that $\frac{d^2 R_{1,\textnormal{bd}} }{d R_{2,\textnormal{bd}}^2}$ is not always non-positive, thus the rate region is not always convex, but it depends on the specific values of $\beta_1$, $\beta_2$, $M$, and  $\rho_{\textnormal{dl}}$.

Consider the special case with $\beta_1=\beta_2$, where $\alpha_1=\alpha_2$ and $\mu_{1,2}=\mu_{2,1}$. The second derivative $\frac{d^2 R_{1,\textnormal{bd}} }{d R_{2,\textnormal{bd}}^2}$ can be simplified and it is non-positive for any $\eta_{1}\in[0,1]$ if and only if 
\begin{equation}
\alpha_1\geq \mu_{1,2}(2+\mu_{1,2}).
\end{equation}
For equal per-antenna SNR, $\mathsf{SNR}\!=\!\rho_{\textnormal{dl}}\beta_1\!=\!\rho_{\textnormal{dl}}\beta_2$,  setting $\alpha_1=M \mathsf{SNR}$ and $\mu_{1,2}=g(\theta_1, \theta_2) \mathsf{SNR}$ gives the convexity condition
\begin{equation}
g(\theta_1,\theta_2)\leq\frac{\sqrt{M\cdot \mathsf{SNR}+1}-1}{\mathsf{SNR}}.
\label{eq:mimi-los-dl}
\end{equation}

Defining the threshold $g^*=\frac{\sqrt{M\cdot \mathsf{SNR}+1}-1}{\mathsf{SNR}}$, for a given pair of $(\theta_1, \theta_2)$, if $g(\theta_1,\theta_2)$ falls above this threshold, the DL achievable rate region is non-convex. 
Thus, for a given SNR, there are certain ranges of angular separation between the two users that leads to a non-convex rate region. 
\begin{figure}[ht]
	\vspace{-0.3cm}
	\centering
	\includegraphics[trim={0cm 0cm  0.5cm 0.2cm},clip, width=0.95\columnwidth]{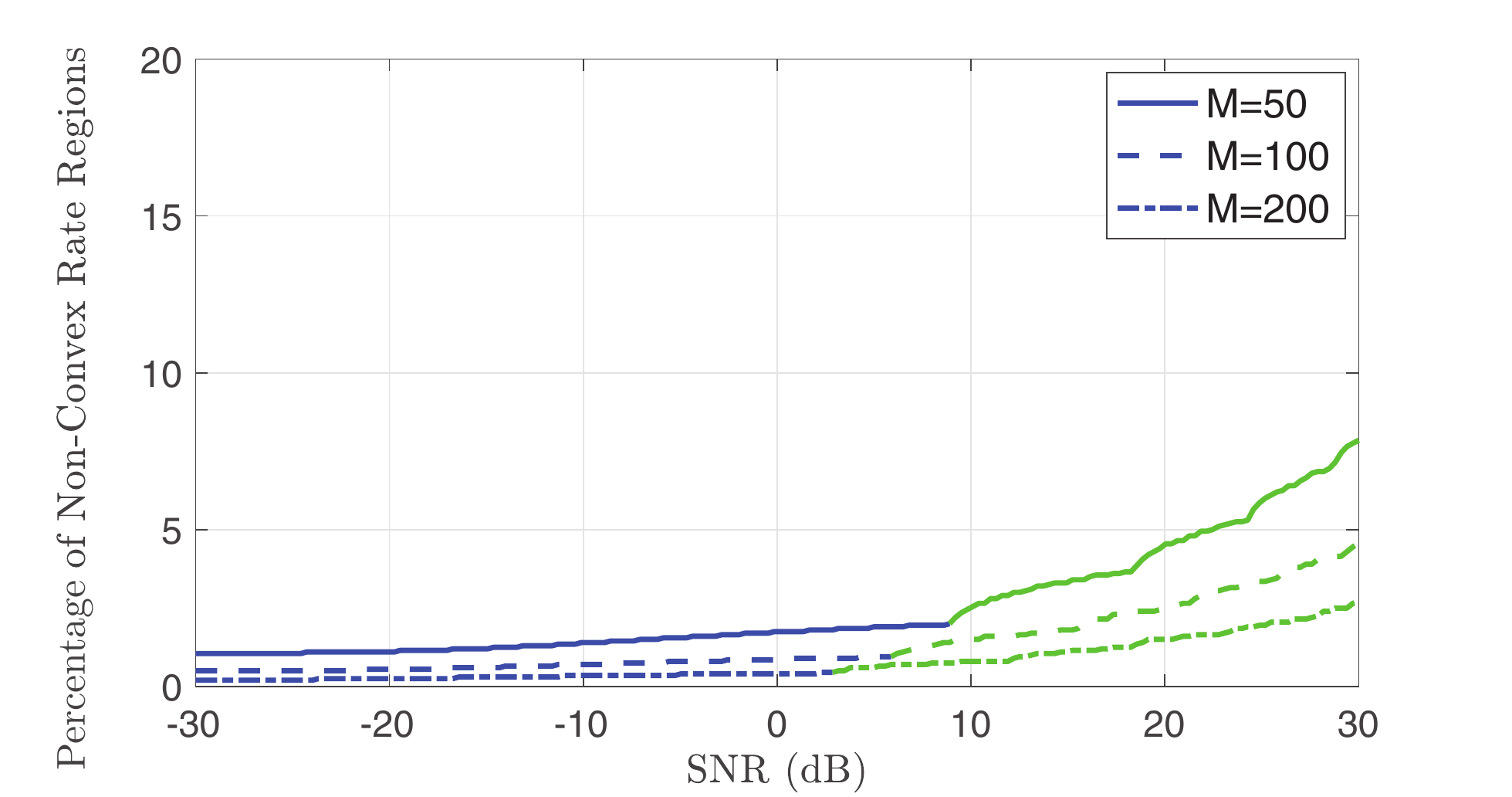}
		\vspace{-0.3cm}
	\caption{Percentage of angles with non-convex rate region when fixing $\theta_1=0$ and varying $\theta_2\in[0, 90\degree]$. $\mathsf{SNR}=\rho_{\textnormal{dl}}\beta_1=\rho_{\textnormal{dl}}\beta_2$. $d_{\textnormal{H}}=0.5$ wavelength.}
	\label{fig:LoS-g-vs-SNR}
\end{figure}

\begin{figure}[ht]
	\vspace{-0.1cm}
	\centering
	\includegraphics[trim={0cm 0cm  0.5cm 0.5cm},clip, width=0.95\columnwidth]{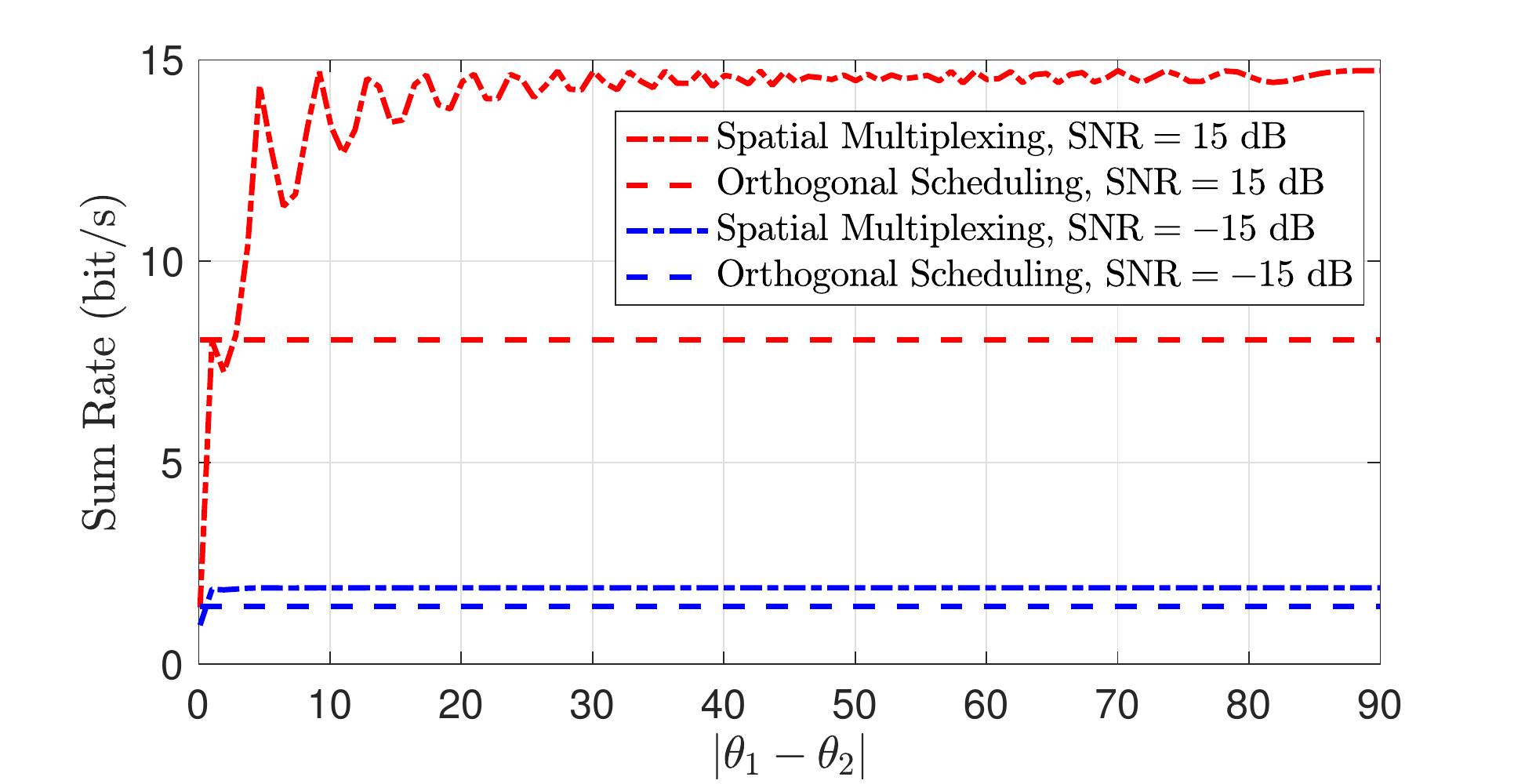}
	\vspace{-0.3cm}
	\caption{Sum rate comparison with $M=100$, same parameters as in Fig.~\ref{fig:LoS-g-vs-SNR}.}
	\label{fig:scheduling-gain}
	\vspace{-0.4cm}
\end{figure}

In Fig.~\ref{fig:LoS-g-vs-SNR}, we fix $\theta_1=0\degree$ and present the percentage of angles $\theta_2\in[0, 90\degree]$ resulting in non-convex rate regions for different SNR values. 
From \eqref{eq:gtheta}, we know that $g(\theta_1,\theta_2)$ does not decrease monotonically with $|\theta_1-\theta_2|$, but it fluctuates when $|\theta_1-\theta_2|$ increases. The green part of the curves means that there is more than one interval of angles that gives non-convex rate regions.
In the low SNR regime, the fraction of angles with non-convex rate region reduces as $1/M$, since the array aperture increases with $M$ (this happens for any $d_{\textnormal{H}}$). In particular, note that $g^*\rightarrow M/2$ when $\mathsf{SNR}\rightarrow 0$ and $g(\theta_1,\theta_2)$ decreases almost linearly with $M$ when $|\theta_1-\theta_2|$ is small.

It is shown in \cite[Ch.~7]{marzetta2016fundamentals} that Massive MIMO with LoS offers nearly favorable propagation, but the channel can become unfavorable if $|\sin(\theta_1)-\sin(\theta_2)|\leq 1/M$. Our observation is that when $|\sin(\theta_1)-\sin(\theta_2)|$ is very small, the rate region becomes non-convex. Both observations suggest that in LoS, we should schedule users with very similar angles orthogonally (e.g., using TDMA, FDMA, or OFDMA). For example, Fig.~\ref{fig:scheduling-gain} shows the sum rate obtained with spatial multiplexing and orthogonal equal-resource scheduling for different $|\theta_1-\theta_2|$. In the extreme case of $\theta_1\simeq\theta_2$, the scheduling gain is $480\%$ when $\text{SNR}=15$\,dB, while for many other angles, spatial multiplexing gives 80\% higher sum rate.
\vspace{-0.3cm}
\subsection{Convexity of the UL Rate Region}
The achievable UL rates with MR combining are given by
$R_{k}\!=\!\log_2\!\left(\!1+\!\frac{M\rho_{\textnormal{ul}}\beta_k\eta_k}{1+g(\theta_k, \theta_j)\rho_{\textnormal{ul}}\beta_j\eta_j}\right)$ for $k,j=\{1,2\}$ and $j\neq k$ \cite{massivemimobook}.
At the first boundary segment with $\eta_{1}=1$, we have
\vspace{-0.1cm}
\begin{equation}
R_{1,\textnormal{bd1}}\!=\!\log_2\!\left(\!1+\!\frac{\alpha_1}{1+\mu_{1,2}\eta_{2}}\right),\,\, R_{2,\textnormal{bd1}}\!=\!\log_2\!\left(\!1+\!\frac{\alpha_2\eta_2}{1+\mu_{2,1}}\right),\\
\end{equation}
where $\alpha_1=M\rho_{\textnormal{ul}}\beta_1$, $\alpha_2=M\rho_{\textnormal{ul}}\beta_2$, $\mu_{1,2}=g(\theta_1, \theta_2)\rho_{\textnormal{ul}}\beta_2$, $\mu_{2,1}=g(\theta_1, \theta_2)\rho_{\textnormal{ul}}\beta_1$.
At the second segment of the boundary with $\eta_{2}=1$, we have
\vspace{-0.1cm}
\begin{equation}
R_{1,\textnormal{bd2}}\!=\!\log_2\!\left(\!1+\!\frac{\alpha_1\eta_{1}}{1+\mu_{1,2}}\right),\,\,R_{2,\textnormal{bd2}}\!=\!\log_2\!\left(\!1+\!\frac{\alpha_2}{1+\mu_{2,1}\eta_{1}}\right).\\
\end{equation}
The convexity conditions in Lemma~\ref{lemma-ul} are not always satisfied, thus the UL rate region is not always convex.

In the special case with $\beta_1=\beta_2$, we have $\alpha_1=\alpha_2$ and $\mu_{1,2}=\mu_{2,1}$.
From Lemma~\ref{def:chain_rule} and Lemma~\ref{lemma-ul}, in order to have $\frac{d^2 R_{1,\textnormal{bd1}}}{d R_{2,\textnormal{bd1}}^2} \leq 0$ for any $\eta_{2}\in[0,1]$, we obtain 
\vspace{-0.2cm}
\begin{equation}
-\alpha_1^2+2\mu_{1,2}(1+\mu_{1,2})(1+\mu_{1,2}\eta_{2})+\alpha_1[-1+\mu_{1,2}+\mu_{1,2}^2(1+\eta_{2}^2)]\leq 0.
\label{eq:condition}
\end{equation}
The left-hand side of \eqref{eq:condition} is a quadratic function of $\eta_{2}$, which monotonically increases with $\eta_{2}$ within the range $0\leq \eta_{2}\leq 1$. Thus, its maximum value is reached at $\eta_{2}=1$. Plugging  $\eta_{2}=1$ into \eqref{eq:condition}, we have 
\vspace{-0.2cm}
\begin{eqnarray}
&2\mu_{1,2}(1+\mu_{1,2})^2+\alpha_1(2\mu_{1,2}^2+\mu_{1,2}-1)-\alpha_1^2\leq 0 \nonumber\\
\Rightarrow& [\alpha_1-2\mu_{1,2}(1+\mu_{1,2})][\alpha_1+(1+\mu_{1,2})]\geq 0\nonumber\\
\Rightarrow& \alpha_1\geq 2\mu_{1,2}(1+\mu_{1,2}).
\end{eqnarray}
Plugging in $\alpha_1=M\cdot \mathsf{SNR}$ and $\mu_{1,2}=g(\theta_1, \theta_2)\cdot \mathsf{SNR}$ with $\mathsf{SNR}=\rho_{\textnormal{ul}}\beta_1=\rho_{\textnormal{ul}}\beta_2$, we obtain 
\begin{equation}
\vspace{-0.2cm}
g(\theta_1, \theta_2)\leq \frac{\sqrt{2M\cdot \mathsf{SNR}+1}-1}{2 \mathsf{SNR}}.
\label{eq:first-two-condition}
\end{equation}
The third condition in Lemma~\ref{lemma-ul} gives $\frac{\mu_{1,2}}{1+\mu_{1,2}}\leq\frac{1+\mu_{1,2}}{\mu_{1,2}}$, which always holds.
Thus, we have that the UL rate region in Massive MIMO LoS is convex if and only if $g(\theta_1, \theta_2)\leq \frac{\sqrt{2M\cdot \mathsf{SNR}+1}-1}{2 \cdot \mathsf{SNR}}$. 

Similar to the DL case, in the low SNR regime, when fixing $\theta_1=0\degree$ and varying $\theta_2\in[0,90\degree]$, the percentage of angles that gives non-convex rate region scales proportionally with $1/M$.

\vspace{-0.15cm}
\begin{propi}
	\label{propo2} 
	In Massive MIMO with LoS channels, the convexity of the two-user rate region depends on the number of antennas $M$, SNRs, and user angles. In the low-SNR regime with $\beta_1=\beta_2$, the probability of having a non-convex region reduces as $1/M$ when the angles are uniformly distributed.
\end{propi}
\vspace{-0.15cm}
\section{Conclusions}
We studied the convexity of the achievable rate region in two-user Massive MIMO systems. We observed that with ergodic i.i.d.~Rayleigh fading, the rate region is always convex, thus it is always beneficial to serve the users by spatial multiplexing (Proposition~\ref{propo1}). With LoS channels, the convexity conditions depend on the number of antennas, the SNR values and the angles of the two users (Proposition~\ref{propo2}). When the angles are similar, the region is non-convex and  scheduling (time-sharing between the corner points) is preferable. The extension to more than two users is important future work.

\vspace{-0.1cm}

\bibliographystyle{IEEEtran}
\bibliography{ref}

\begin{thebibliography}{1}
\providecommand{\url}[1]{#1}
\csname url@samestyle\endcsname
\providecommand{\newblock}{\relax}
\providecommand{\bibinfo}[2]{#2}
\providecommand{\BIBentrySTDinterwordspacing}{\spaceskip=0pt\relax}
\providecommand{\BIBentryALTinterwordstretchfactor}{4}
\providecommand{\BIBentryALTinterwordspacing}{\spaceskip=\fontdimen2\font plus
\BIBentryALTinterwordstretchfactor\fontdimen3\font minus
  \fontdimen4\font\relax}
\providecommand{\BIBforeignlanguage}[2]{{%
\expandafter\ifx\csname l@#1\endcsname\relax
\typeout{** WARNING: IEEEtran.bst: No hyphenation pattern has been}%
\typeout{** loaded for the language `#1'. Using the pattern for}%
\typeout{** the default language instead.}%
\else
\language=\csname l@#1\endcsname
\fi
#2}}
\providecommand{\BIBdecl}{\relax}
\BIBdecl

\bibitem{caire-achievable-mimo}
G.~Caire, N.~Jindal, M.~Kobayashi, and N.~Ravindran, ``Multiuser {MIMO}
  achievable rates with downlink training and channel state feedback,''
  \emph{IEEE Trans. on Information Theory}, vol.~56, no.~6, pp. 2845--2866,
  June 2010.

\bibitem{ergodic-mimo-broadcast}
A.~Hindy and A.~Nosratinia, ``Ergodic fading {MIMO} dirty paper and broadcast
  channels: Capacity bounds and lattice strategies,'' \emph{IEEE Trans. on
  Wireless Communications}, vol.~16, no.~8, pp. 5525--5536, Aug 2017.

\bibitem{bjornson2013optimal}
E.~Bj{\"o}rnson and E.~Jorswieck, ``Optimal resource allocation in coordinated
  multi-cell systems,'' \emph{Foundations and Trends{\textregistered} in
  Communications and Information Theory}, vol.~9, no. 2--3, pp. 113--381, 2013.

\bibitem{capacity-mimo-broadcast}
H.~Weingarten, Y.~Steinberg, and S.~S. Shamai, ``The capacity region of the
  {Gaussian} multiple-input multiple-output broadcast channel,'' \emph{IEEE
  Trans. on Information Theory}, vol.~52, no.~9, pp. 3936--3964, Sept 2006.

\bibitem{zhang-overview}
L.~Lu, G.~Y. Li, A.~L. Swindlehurst, A.~Ashikhmin, and R.~Zhang, ``An overview
  of {Massive MIMO}: Benefits and challenges,'' \emph{IEEE Journal of Sel.
  Topics in Signal Processing}, vol.~8, no.~5, pp. 742--758, Oct 2014.

\bibitem{massivemimobook}
E.~Bj\"{o}rnson, J.~Hoydis, and L.~Sanguinetti, ``Massive {MIMO} networks:
  {Spectral}, energy, and hardware efficiency,'' \emph{Foundations and
  Trends{\textregistered} in Signal Processing}, vol.~11, no. 3-4, pp.
  154--655, 2017.

\bibitem{capacity-region}
J.~Diakonikolas and G.~Zussman, ``On the rate regions of single-channel and
  multi-channel full-duplex links,'' \emph{IEEE/ACM Transactions on
  Networking}, vol.~26, no.~1, pp. 47--60, 2018.

\bibitem{marzetta2016fundamentals}
T.~L. Marzetta, E.~G. Larsson, H.~Yang, and H.~Q. Ngo, \emph{Fundamentals of
  {Massive MIMO}}.\hskip 1em plus 0.5em minus 0.4em\relax Cambridge University
  Press, 2016.

\end{thebibliography}

\end{document}